\documentclass[11pt,letterpaper]{article}
\usepackage{times,mathptmx}
\DeclareMathAlphabet{\mathcal}{OMS}{cmsy}{m}{n}
\usepackage{fullpage}
\usepackage{amsmath,amsfonts,amssymb,amsthm,amscd}
\usepackage{tabularx,booktabs}
\usepackage{hyperref,cite,url}
\hypersetup{pdfpagemode=UseNone,pdfstartview=}
\theoremstyle{plain} 
\newtheorem{theorem}{Theorem}[section]
\newtheorem{lemma}[theorem]{Lemma}

\theoremstyle{definition} 
\newtheorem{definition}{Definition}[section]
\newtheorem{assumption}{Assumption}
\theoremstyle{remark} 
\newtheorem{remark}{Remark}
\newcommand{\svs}{\vspace{0.7mm}}
\newcommand{\vs}{\vspace{1.5mm}}
\newcommand{\G}{\mathbb{G}}
\newcommand{\Z}{\mathbb{Z}}
\newcommand{\bits}{\{0,1\}}
\newcommand{\mc}[1]{\mathcal{#1}}
\newcommand{\tb}[1]{\textbf{#1}}
\newcommand{\ts}[1]{\textsf{#1}}
\newcommand{\lb}{\linebreak[0]}

\newcommand{\Adv}{\textbf{Adv}}

\title{Adaptive Chosen-Ciphertext Security of\\ Distributed Broadcast Encryption}

\author{
    Kwangsu Lee\footnote{Sejong University, Seoul, Korea.
        Email: \texttt{kwangsu@sejong.ac.kr}.}
}

\date{}

\begin{document}

\maketitle

\begin{abstract}
Distributed broadcast encryption (DBE) is a specific kind of broadcast
encryption (BE) where users independently generate their own public and
private keys, and a sender can efficiently create a ciphertext for a subset
of users by using the public keys of the subset users. Previously proposed
DBE schemes have been proven in the adaptive chosen-plaintext attack (CPA)
security model and have the disadvantage of requiring linear number of
pairing operations when verifying the public key of a user. In this paper,
we propose an efficient DBE scheme in bilinear groups and prove adaptive
chosen-ciphertext attack (CCA) security for the first time. To do this, we
first propose a semi-static CCA secure DBE scheme and prove the security
under the $q$-Type assumption. Then, by modifying the generic
transformation of Gentry and Waters that converts a semi-static CPA secure
DBE scheme into an adaptive CPA secure DBE scheme to be applied to CCA
secure DBE schemes, we propose an adaptive CCA secure DBE scheme and prove
its adaptive CCA security. Our proposed DBE scheme is efficient because it
requires constant size ciphertexts, constant size private keys, and linear
size public keys, and the public key verification requires only a constant
number of pairing operations and efficient group membership checks.
\end{abstract}

\vs \noindent {\bf Keywords:} Distributed broadcast encryption,
Chosen-ciphertext security, Adaptive security, Bilinear maps.

\newpage

\section{Introduction}

Broadcast encryption (BE) is an encryption method that efficiently transmits
an encrypted message to a subset of users, and the concept of BE was first
introduced by Fiat and Naor \cite{FiatN93}. In a general public-key BE
scheme, the private keys of users are generated by a trusted central
authority and transmitted securely to users. Afterwards, a sender creates a
ciphertext for the subset of recipients by using the public parameters, and a
recipient can decrypt the ciphertext with his or her private key if his or
her index is included in the subset. By combining the ciphertexts of
public-key encryption (PKE), it is possible to design a trivial BE scheme
with a ciphertext size that increases linearly with the size of the recipient
subset. Thus, an efficient BE scheme requires that the ciphertext has a
sub-linear size. The most efficient public-key BE scheme is the BGW-BE scheme
designed by Boneh et al. \cite{BonehGW05} in bilinear groups with constant
size ciphertexts and constant size private keys, which provides static
security under the $q$-Type assumption. Since then, various BE schemes have
been proposed in bilinear groups, multilinear maps, indistinguishability
obfuscation, and lattices \cite{Delerablee07,GentryW09,Waters09, BonehWZ14,
BonehZ14m,GayKW18,AgrawalY20,Wee22}.

Traditional PKE schemes do not require a trusted center for private key
generation because individual users independently generate public and private
keys. Since BE schemes require a trusted center to generate user private
keys, the trust model of BE must be stronger than that of PKE. This makes it
difficult to widely apply BE schemes to large-scale real-world applications
where only the weak trust model is applied. Recently, the concept of
distributed broadcast encryption (DBE) has been proposed \cite{WuQZD10,
BonehZ14m}, which does not require a trusted center for user private key
generation. In DBE schemes, individual users generate their own public and
private keys and store the public keys in a public repository. Then, a sender
creates a ciphertext using the public keys for a subset of receivers, and a
receiver can decrypt the ciphertext if he or she belongs to the subset of
receivers in the ciphertext. Recently, Kolonelos et al. proposed efficient
DBE schemes by decentralizing the private key generation of efficient BE
schemes in bilinear groups and proved the adaptive CPA security of their DBE
schemes \cite{KolonelosMW23}.

The standard security model for PKE is the CCA security model, which allows
an attacker to request ciphertext decryption queries \cite{RackoffS91}. The
right security model for DBE schemes is also the adaptive CCA security model,
which allows an attacker to submit a challenge target set in the challenge
phase and request ciphertext decryption queries. Previously proposed DBE
schemes have considered adaptive CPA security, but not adaptive CCA security.
In this paper, we examine the problem of designing a DBE scheme that provides
adaptive CCA security.

\subsection{Our Contributions}

To design an adaptive CCA secure DBE scheme, we first propose a semi-static
CCA secure DBE scheme. A semi-static CCA security model is a restricted model
in which an attacker initially commits an initial set $\tilde{S}$ and submits
a challenge target set $S^* \subseteq \tilde{S}$ in the challenge phase,
allowing the attacker to request key generation and decryption queries. To
construct our DBE$_{SS}$ scheme, we start from the KMW-DBE1 scheme of
Kolonelos et al. \cite{KolonelosMW23} that is derived from the BGW-BE scheme
of Boneh et al. \cite{BonehGW05} that provides CPA security. In this case, we
directly use a ciphertext element to provide ciphertext integrity instead of
using a signature scheme to improve efficiency, and we modify the decryption
process to derive a randomized private key like the private key of IBE. In
addition, we apply a batch verification method to check public key elements
of the user's public key. Since a simple batch verification method does not
work, we change one group element in a ciphertext from $\G_1$ to $\G_2$ and
set most of the public key elements to be group elements in $\G_1$. Our
proposed DBE$_{SS}$ scheme has constant size ciphertexts, constant size
private keys, linear size public keys, and linear size public parameters. The
public key verification of our scheme is very efficient since it only
requires two pairing operations and efficient group membership check
operations. We prove the semi-static CCA security of our DBE$_{SS}$ scheme
under the $q$-Type assumption and the collision resistance of hash functions.

Next, we design an adaptive CCA secure DBE$_{AD}$ scheme by combining the
semi-static CCA secure DBE$_{SS}$ scheme designed above, a strongly
unforgeable one-time signature (OTS) scheme, and the GW transformation of
Gentry and Waters \cite{GentryW09}, and prove adaptive CCA security through
the hybrid games. In an adaptive CCA security model, an attacker can request
key generation, key corruption, and decryption queries, and the attacker
submits a challenge target set $S^*$ in the challenge phase and obtains the
corresponding challenge ciphertext header and challenge session key. The
attacker wins this game if he or she can distinguish whether the challenge
session key is correct or random. The GW transformation is a method that
converts a semi-static CPA secure DBE scheme into an adaptive CPA secure DBE
scheme by doubling the ciphertext size. We modify the existing GW
transformation to combine a semi-static CCA secure DBE scheme and a strongly
unforgeable OTS scheme for CCA security as well. Our proposed DBE$_{AD}$
scheme is the first DBE scheme that provides adaptive CCA security, and it
also naturally satisfies active-adaptive CCA security in which an attacker
registers malicious public keys. The comparison of our DBE scheme with the
previous DBE schemes is given in Table \ref{tab:comp-dbe}.

\begin{table*}[t]
\caption{Comparison of DBE schemes in bilinear groups}
\label{tab:comp-dbe}
\vs \small \addtolength{\tabcolsep}{8.0pt}
\renewcommand{\arraystretch}{1.4}
\newcommand{\otoprule}{\midrule[0.09em]}
\begin{tabularx}{6.50in}{lccccccc}
\toprule
Scheme  & Type & PP  & USK  & UPK  & CT  & Model & Assumption \\
\otoprule
WQZD \cite{WuQZD10}
    & DBE & $O(L)$ & $O(L)$ & $O(L^2)$ & $O(1)$	& AD-CPA & $q$-Type \\
KMW \cite{KolonelosMW23}
    & DBE & $O(L)$ & $O(1)$ & $O(L)$ & $O(1)$   & AD-CPA & $q$-Type \\
KMW \cite{KolonelosMW23}
    & DBE & $O(L^2)$ & $O(1)$ & $O(L)$ & $O(1)$ & AD-CPA & $k$-LIN \\
Lee \cite{Lee25}
    & DBE & $O(L)$ & $O(1)$ & $O(L)$ & $O(1)$   & AD-CPA & SD, GSD \\
\midrule[0.05em]
Ours & DBE & $O(L)$ & $O(1)$ & $O(L)$ & $O(1)$   & AD-CCA & $q$-Type \\
\bottomrule
\multicolumn{8}{p{6.10in}}{
Let $L$ be the number of all users.
We count the number of group elements to measure the size.
We use symbols AD-CPA for adaptive CPA security and AD-CCA for adaptive CCA
security.
}
\end{tabularx}
\end{table*}

\subsection{Related Work}

\tb{Chosen-Ciphertext Security.} The right definition of CCA security in PKE
was given by Rackoff and Simon and they showed that a CCA secure PKE scheme
can be constructed by combining a CPA secure PKE scheme with a
non-interactive zero-knowledge proof (NIZK) system \cite{RackoffS91}. Cramer
and Shoup proposed an efficient CCA secure PKE scheme by using a hash proof
system (HPS), which is a special form of zero-knowledge proof
\cite{CramerS02,CramerS03}. Canetti et al. presented a transformation method
that converts a CPA secure IBE scheme to a CCA secure PKE scheme by using an
one-time signature (OTS) scheme for ciphertext integrity \cite{CanettiHK04}.
After that, Boneh and Katz showed that a more efficient CCA secure PKE scheme
can be built by using a message authentication code (MAC) scheme instead of
the OTS scheme in the CHK transformation \cite{BonehCHK07}. Boyen et al.
modified the CHK transformation method and showed that an efficient CCA
secure KEM without additional OTS and MAC schemes can be built by using the
IBE scheme directly \cite{BoyenMW05}. Dodis and Katz proposed a method to
ensure CCA security when a ciphertext consists of multiple independent
ciphertexts \cite{DodisK05}.

\vs\noindent \tb{Broadcast Encryption.} The concept of BE was first
introduced by Fiat and Naor and they proposed a BE scheme by using
combinatorial methods \cite{FiatN93}. Naor et al. proposed BE schemes that
are secure against collusion attacks without maintaining state information by
using binary trees \cite{NaorNL01}. Boneh et al. proposed the first efficient
BE scheme with constant size ciphertexts in bilinear groups and proved static
CPA security under the $q$-Type assumption \cite{BonehGW05}. In addition,
they proposed a CCA secure BE scheme by combining the IBE private key
generation method with an OTS scheme. Abdalla et al. proposed a method to
design a BE scheme with constant size ciphertexts by extending the private
key delegation function of the hierarchical identity-based encryption (HIBE)
scheme with constant size ciphertexts \cite{AbdallaKN07}. Most of the
existing BE schemes are proven in the static model that specifies the
challenge target set in advance, but the right security model is the adaptive
model that specifies the challenge target set in the challenge phase. Gentry
and Waters presented an efficient transformation to design an adaptive CPA
secure BE scheme by doubling the ciphertext size of a semi-static CPA secure
BE scheme \cite{GentryW09}. Recently, methods for designing efficient BE
schemes with optimal parameters and ciphertext sizes using attribute-based
encryption (ABE) schemes have been proposed \cite{AgrawalY20,AgrawalWY20,
Wee22}.

\vs\noindent \tb{Distributed Broadcast Encryption.} Wu et al. introduced the
concept of ad-hoc broadcast encryption (AHBE) in which individual users
independently generate private and public keys, and then proposed an AHBE
scheme with constant size ciphertexts and linear size private keys in
bilinear groups \cite{WuQZD10}. Boneh and Zhandry defined the concept of DBE
and showed that a DBE scheme can be designed using indistinguishable
obfuscation \cite{BonehZ14m}. DBE schemes require user indexes when
generating private keys of users, but flexible broadcast encryption (FBE)
schemes do not require user indexes for key generation. Garg et al. proposed
a general method to convert a DBE scheme into an FBE scheme by increasing the
size of users' public keys \cite{GargLWW23}. Kolonelos et al. proposed
efficient DBE schemes with constant size ciphertexts and constant size
private keys in bilinear groups and proved adaptive CPA security
\cite{KolonelosMW23}. Champion and Wu designed the first DBE scheme in
lattices and proved the static CPA security \cite{ChampionW24}.

\section{Preliminaries}

In this section, we review bilinear groups with complexity assumptions,
symmetric-key encryption, and strongly unforgeable one-time signatures.

\subsection{Bilinear Groups and Complexity Assumptions}

A bilinear group generator $\mc{G}$ takes as input a security parameter
$\lambda$ and outputs a tuple $(p, \G, \hat{\G}, \G_T, e)$ where $p$ is a
random prime and $\G, \hat{\G}$, and $\G_T$ be three cyclic groups of prime
order $p$. Let $g$ and $\hat{g}$ be generators of $\G$ and $\hat{\G}$,
respectively. The bilinear map $e : \G \times \hat{\G} \rightarrow \G_{T}$
has the following properties:
\begin{enumerate}
\item Bilinearity: $\forall u \in \G, \forall \hat{v} \in \hat{\G}$ and
    $\forall a,b \in \Z_p$, $e(u^a,\hat{v}^b) = e(u,\hat{v})^{ab}$.
\item Non-degeneracy: $\exists g \in \G, \hat{g} \in \hat{\G}$ such that
    $e(g,\hat{g})$ has order $p$ in $\G_T$.
\end{enumerate}
We say that $\G, \hat{\G}, \G_T$ are asymmetric bilinear groups with no
efficiently computable isomorphisms if the group operations in $\G,
\hat{\G}$, and $\G_T$ as well as the bilinear map $e$ are all efficiently
computable, but there are no efficiently computable isomorphisms between $\G$
and $\hat{\G}$.

\begin{assumption}[Bilinear Diffie-Hellman Exponent, BDHE \cite{BonehGW05}]
Let $(p, \G, \hat{\G}, \G_T, e)$ be an asymmetric bilinear group generated by
$\mc{G}(1^\lambda)$. Let $g, \hat{g}$ be random generators of $\G, \hat{\G}$
respectively. The $\ell$-bilinear Diffie-Hellman exponent ($\ell$-BDHE)
assumption is that if the challenge tuple
	$$D = \big( (p, \G, \hat{\G}, \G_T, e), g, g^a, \ldots, g^{a^\ell},
	g^c, \hat{g}, \hat{g}^a, \ldots, \hat{g}^{a^\ell}, \hat{g}^{a^{\ell+2}},
	\ldots, \hat{g}^{a^{2\ell}} \big) \mbox{ and } Z$$
are given, no PPT algorithm $\mc{A}$ can distinguish $Z = Z_0 =
e(g,\hat{g})^{a^{\ell+1}c}$ from $Z = Z_1 = e(g,\hat{g})^d$ with more than a
negligible advantage. The advantage of $\mc{A}$ is defined as
    $\Adv_{\mc{A}}^{BDHE} (\lambda) = \big|
    \Pr[\mc{A}(D,Z_0) = 0] - \Pr[\mc{A}(D,Z_1) = 0] \big|$
where the probability is taken over random choices of $a, c, d \in \Z_p$.
\end{assumption}

\subsection{Symmetric Key Encryption}

Symmetric key encryption (SKE) is a cryptographic technique that uses the
same symmetric key for encryption and decryption algorithms. The traditional
security model for SKE is the indistinguishability under chosen-plaintext
attack (IND-CPA) model, but we define the one-message indistinguishability
(OMI) model that considers a single challenge ciphertext. If an SKE scheme is
IND-CPA secure, then it is also OMI secure.

\begin{definition}[Symmetric Key Encryption]
A symmetric key encryption (SKE) scheme consists of three algorithms
$\tb{GenKey}, \tb{Encrypt}$, and $\tb{Decrypt}$, which are defined as follows:
\begin{description}
\item $\tb{GenKey}(1^\lambda)$: The key generation algorithm takes as input a
security parameter $\lambda$. It outputs a symmetric key $K$.

\item $\tb{Encrypt}(K, M)$: The encryption algorithm takes as input a
symmetric key $K$ and a message $M$. It outputs a ciphertext $C$.

\item $\tb{Decrypt}(K, C)$: The decryption algorithm takes as input a
symmetric key $K$ and a ciphertext $C$. It outputs a message $M$ or a special
symbol $\perp$.
\end{description}
The correctness property of SKE is defined as follows: For all $K$ generated
by $\tb{GenKey}(1^\lambda)$ and any message $M$, it is required that
$\tb{Decrypt} (K, \tb{Encrypt}(K, M)) = M$.
\end{definition}

\begin{definition}[One-Message Indistinguishability] \label{def:ske-omi-sec}
The one-message indistinguishability (OMI) of SKE is defined in terms of the
following experiment between a challenger $\mc{C}$ and a PPT adversary $\mc{A}$
where $1^\lambda$ is given as input:
\begin{enumerate}
\item \tb{Setup}: $\mc{C}$ obtains a symmetric key $K$ by running $\tb{GenKey}
(1^{\lambda})$ and keeps $K$ to itself.

\item \tb{Challenge}: $\mc{A}$ submits challenge messages $M_0^*, M_1^*$ where
$|M_0^*| = |M_1^*|$. $\mc{C}$ flips a random coin $\mu \in \bits$ and obtains
$CT^*$ by running $\tb{Encrypt}(K, M_\mu^*)$. It gives $CT^*$ to $\mc{A}$.

\item \tb{Guess}: $\mc{A}$ outputs a guess $\mu' \in \bits$. $\mc{C}$ outputs
$1$ if $\mu = \mu'$ or $0$ otherwise.
\end{enumerate}
The advantage of $\mc{A}$ is defined as $\Adv_{SKE,\mc{A}}^{OMI} (\lambda) =
\big| \Pr[\mu = \mu'] - \frac{1}{2} \big|$ where the probability is taken over
all the randomness of the experiment. An SKE scheme is OMI secure if for all
probabilistic polynomial-time (PPT) adversary $\mc{A}$, the advantage of
$\mc{A}$ is negligible in the security parameter $\lambda$.
\end{definition}

\subsection{One-Time Signature}

One-time signature (OTS) is a special kind of public-key signature (PKS) that
allows an attacker to obtain at most one signature. The security model of OTS
is strong unforgeability, which allows an attacker to obtain at most one
signature from a signing oracle, and the message queried to the signing
oracle is also allowed to be a forged message if the forged signature by the
attacker is different from a signature received from the signing oracle.

\begin{definition}[One-Time Signature]
A one-time signature (OTS) scheme consists of three algorithms $\tb{GenKey},
\tb{Sign}$, and $\tb{Verify}$, which are defined as follows:
\begin{description}
\item $\tb{GenKey}(1^\lambda)$: The key generation algorithm takes as input a
security parameter $\lambda$. It outputs a signing key $SK$ and a verification
key $VK$.

\item $\tb{Sign}(SK, M)$: The signing algorithm takes as input a signing key
$SK$ and a message $M$. It outputs a signature $\sigma$.

\item $\tb{Verify}(VK, \sigma, M)$: The verification algorithm takes as input
a verification key $VK$, signature $\sigma$, and a message $M$. It outputs $1$
if the signature is valid and $0$ otherwise.
\end{description}
The correctness property of OTS is defined as follows: For all $(SK, VK)$
generated by $\tb{GenKey}(1^\lambda)$ and any message $M$, it is required that
$\tb{Verify} (VK, \tb{Sign}(SK, M), M) = 1$.
\end{definition}

\begin{definition}[Strong Unforgeability] \label{def:ots-suf-sec}
The strong unforgeability (SUF) of OTS is defined in terms of the following
experiment between a challenger $\mc{C}$ and a PPT adversary $\mc{A}$ where
$1^\lambda$ is given as input:
\begin{enumerate}
\item \tb{Setup}: $\mc{C}$ first generates a key pair $(SK, VK)$ by running
$\tb{GenKey}(1^{\lambda})$ and gives $VK$ to $\mc{A}$.

\item \tb{Signature Query}: $\mc{A}$ requests at most one signature query on
a message $M$. $\mc{C}$ generates a signature $\sigma$ by running $\tb{Sign}
(SK, M)$ and gives $\sigma$ to $\mc{A}$.

\item \tb{Output}: Finally, $\mc{A}$ outputs a forged pair $(\sigma^*, M^*)$.
$\mc{C}$ outputs $1$ if the forged pair satisfies the following conditions,
or outputs $0$ otherwise: 1) $\tb{Verify}(VK, \sigma^*, M^*) = 1$,
2) $(\sigma^*, M^*) \neq (\sigma, M)$ where $(\sigma, M)$ is the pair of
the signature query.
\end{enumerate}
The advantage of $\mc{A}$ is defined as $\Adv_{OTS,\mc{A}}^{SUF} (\lambda) =
\Pr[\mc{C} = 1]$ where the probability is taken over all the randomness of the
experiment. An OTS scheme is SUF secure if for all probabilistic
polynomial-time (PPT) adversary $\mc{A}$, the advantage of $\mc{A}$ is
negligible in the security parameter $\lambda$.
\end{definition}

Efficient PKS schemes that provide strong unforgeability in bilinear groups
include the BLS-PKS scheme of Boneh et al. \cite{BonehLS04} in the random
oracle model and the BB-PKS scheme of Boneh and Boyen \cite{BonehB04s}
without the random oracle model.

\section{Distributed Broadcast Encryption}

In this section, we define the syntax of DBE and the security model of DBE.

\subsection{Definition}

In a DBE scheme, a trusted center runs the setup algorithm to generate public
parameters. Each user generates his or her own private key and public key by
running the key generation algorithm and discloses the public key to a public
directory. A sender runs the encryption algorithm to create a ciphertext
header and a session key using a subset of recipients and their public keys.
After that, a receiver can derive the same session key by running the
decryption algorithm using the recipients' public keys and its own private
key If the index of a receiver is included in the subset of the ciphertext
header. A more detailed syntax of a DBE scheme is given as follows:

\begin{definition}[Distributed Broadcast Encryption]
A distributed broadcast encryption (DBE) scheme consists of five algorithms
$\tb{Setup}, \tb{GenKey}, \tb{IsValid}, \tb{Encaps}$, and $\tb{Decaps}$,
which are defined as follows:
\begin{description}
\item $\tb{Setup}(1^\lambda, 1^L)$: This algorithm takes as input a security
parameter $1^\lambda$, and the number users $L$. It outputs public parameters
$PP$.

\item $\tb{GenKey}(i, PP)$: This algorithm takes as input a user index $i \in
[L]$ and public parameters $PP$. It outputs a private key $USK_i$ and a public
key $UPK_i$.

\item $\tb{IsValid}(j, UPK_j, PP)$: This algorithm takes as input an index
$j$, a public key $UPK_j$, and the public parameters $PP$. It outputs $1$ or
$0$ depending on the validity of keys.

\item $\tb{Encaps}(S, \{ (j, UPK_j) \}_{j \in S}, PP, AU)$: This algorithm
takes as input a set $S \subseteq [L]$, public keys $\{ (j, UPK_j) \}_{j \in
S}$, public parameters $PP$, and an optional auxiliary input $AU$. It outputs
a ciphertext header $CH$ and a session key $CK$.

\item $\tb{Decaps}(S, CH, i, USK_i, \{ (j, UPK_j) \}_{j \in S}, PP, AU)$: This
algorithm takes as input a set $S$, a ciphertext header $CH$, an index $i$, a
private key $USK_i$ for the index $i$, public keys $\{ (j, UPK_j) \}_{j \in
S}$, public parameters $PP$, and an optional auxiliary input $AU$. It outputs
a session key $CK$.
\end{description}
The correctness of DBE is defined as follows: For all $PP$ generated by
$\tb{Setup}(1^{\lambda}, 1^L)$, all $(USK_i, UPK_i)$ generated by
$\tb{GenKey}(i, PP)$, all $UPK_j$ such that $\tb{IsValid}(j, UPK_j, PP)$, all
$S \subseteq [L]$, it is required that
\begin{itemize}
\item If $i \in S$, then $CK = CK'$ where $(CH, CK) = \tb{Encaps}(S, \{ (j,
UPK_j) \}_{j \in S}, PP, AU)$ and $CK' = \tb{Decaps}(S, \lb CH, i, USK_i, \{ (j,
UPK_j) \}_{j \in S}, PP, AU)$.
\end{itemize}
\end{definition}

\begin{remark}
We modify the syntax of DBE to accept an auxiliary input $AU$ to the
encryption and decryption algorithms. This $AU$ field is optional and served
as a label.
\end{remark}

\subsection{Security Model}

The semi-static CPA security of BE was defined by Gentry and Waters
\cite{GentryW09}. We extend the existing semi-static CPA security model to
the semi-static CCA security model by adding a decryption oracle. In the
semi-static CCA security model, an attacker first submits an initial set
$\tilde{S}$, and a challenger generates public parameters by running the
setup algorithm. In the query phase, the attacker can request key generation
and decryption queries with some constrains, and the challenger processes
these queries by executing the key generation algorithm or the decryption
algorithm. In the challenge phase, the attacker submits a set $S^* \subseteq
\tilde{S}$, the challenger obtains a ciphertext header $CH^*$ and a session
key $CK^*$ using the encryption algorithm, and it sets $CK_0^* = CK^*$ and
$CK_1^*$ with random. After that, the challenger flips a random coin $\mu$
and sends the challenge $CH^*, CK_{\mu}^*$ to the attacker. Afterwards, the
attacker can additionally request decryption queries except the challenge
ciphertext header. Finally, the attacker wins this game if it correctly
guesses the random coin of the challenger. The detailed definition of the
security model is given as follows:

\begin{definition}[Semi-Static CCA Security] \label{def:dbe-sscca-sec}
The semi-static CCA (SS-CCA) security of DBE is defined in terms of the
following experiment between a challenger $\mc{C}$ and a PPT adversary
$\mc{A}$ where $1^\lambda$ and $1^L$ are given as input:
\begin{enumerate}
\item \tb{Init}: $\mc{A}$ initially commits an initial set $\tilde{S}
\subseteq [L]$.

\item \tb{Setup}: $\mc{C}$ obtains public parameters $PP$ by running
$\tb{Setup}(1^{\lambda}, 1^L)$ and gives $PP$ to $\mc{A}$.

\item \tb{Query Phase 1}: $\mc{A}$ adaptively requests key generation and
decryption queries. These queries are processed as follows:
	\begin{itemize}
	\item Key Generation: For all $j \in \tilde{S}$, $\mc{C}$ generates
	$(USK_j, UPK_j)$ by running $\tb{GenKey} (j, PP)$ and gives $\{ (j, UPK_j)
	\}_{j \in \tilde{S}}$ to $\mc{A}$.

	\item Decryption: $\mc{A}$ issues this query on $(S, CH, i, AU)$ such that
	$S \subseteq \tilde{S}$ and $i \in S$. $\mc{C}$ responds with $\tb{Decaps}
	(S, CH, i, USK_i, \{ (j, UPK_j) \}_{j \in S}, PP, AU)$.
	\end{itemize}

\item \tb{Challenge}: $\mc{A}$ submits a challenge set $S^* \subseteq
\tilde{S}$. $\mc{C}$ obtains a ciphertext tuple $(CH^*, CK^*)$ by running
$\tb{Encaps} (S^*, \{ (j, UPK_j) \}_{j \in S^*}, PP, -)$. It sets $CK_0^* =
CK^*$ and $CK_1^* = RK$ by selecting a random $RK$. It flips a random coin
$\mu \in \bits$ and gives $(CH^*, CK_\mu^*)$ to $\mc{A}$.

\item \tb{Query Phase 2}: $\mc{A}$ continues to request decryption queries.
These queries are processed as follows:
	\begin{itemize}
	\item Decryption: $\mc{A}$ issues this query on $(S, CH, i, AU)$ such that
	$S \subseteq \tilde{S}$, $i \in S$, and $CH \neq CH^*$. $\mc{C}$
	responds with $\tb{Decaps} (S, CH, i, USK_i, \{ (j, UPK_j) \}_{j \in S},
	PP, AU)$.
	\end{itemize}

\item \tb{Guess}: Finally, $\mc{A}$ outputs a guess $\mu' \in \bits$, and wins
the game if $\mu = \mu'$.
\end{enumerate}
The advantage of $\mc{A}$ is defined as $\Adv_{DBE,\mc{A}}^{SS\text{-}CCA}
(\lambda) = \big| \Pr[\mu = \mu'] - \frac{1}{2} \big|$ where the probability
is taken over all the randomness of the experiment. A DBE scheme is SS-CCA
secure if for all probabilistic polynomial-time (PPT) adversary $\mc{A}$, the
advantage of $\mc{A}$ is negligible in the security parameter $\lambda$.
\end{definition}

Adaptive CPA security of DBE is a relaxed security model in which an attacker
can freely specify a challenge set $S^*$ in the challenge phase. We define
the adaptive CCA security model by adding a decryption oracle to the adaptive
CPA security model. In the adaptive CCA security model, unlike the
semi-static CCA model defined above, the attacker does not initially submit
an initial set $\tilde{S}$. In the query phase, the challenger manages a key
generation query set $KQ$, a key corruption query set $CQ$, and a decryption
query set $DQ$ for the queries requested by the attacker. In the challenge
phase, the attacker submits a challenge set $S^* \subseteq KQ \setminus CQ$,
and the challenge ciphertext header and the session key are transmitted to
the attacker. After that, the attacker additionally requests decryption
queries and wins the game if it can correctly guess the coin selected by the
challenger. The detailed definition of the security model is given as
follows:

\begin{definition}[Adaptive CCA Security] \label{def:dbe-adcca-sec}
The adaptive CCA (AD-CCA) security of DBE is defined in terms of the following
experiment between a challenger $\mc{C}$ and a PPT adversary $\mc{A}$ where
$1^\lambda$ and $1^L$ are given as input:
\begin{enumerate}
\item \tb{Setup}: $\mc{C}$ obtains public parameters $PP$ by running
$\tb{Setup} (1^{\lambda}, 1^L)$ and gives $PP$ to $\mc{A}$. It prepares $KQ,
CQ$, and $DQ$ as empty set.

\item \tb{Query Phase 1}: $\mc{A}$ adaptively requests key generation, key
corruption, and decryption queries. These queries are processed as follows:
    \begin{itemize}
	\item Key Generation: $\mc{A}$ issues this query on an index $i \in [L]$
	such that $i \not\in KQ$. $\mc{C}$ creates $(USK_i, UPK_i)$ by running
	$\tb{GenKey}(i, PP)$, adds $i$ to $KQ$, and responds $UPK_i$ to $\mc{A}$.

	\item Key Corruption: $\mc{A}$ issues this query on an index $i \in [L]$
	such that $i \in KQ \setminus (CQ \cup DQ)$. $\mc{C}$ adds $i$ to $CQ$ and
	responds $USK_i$ to $\mc{A}$.

	\item Decryption: $\mc{A}$ issues this query on $(S, CH, i, AU)$ such that
	$S \subseteq KQ \setminus CQ$ and $i \in S$. $\mc{C}$ obtains $CK$ by
	running $\tb{Decaps} (S, CH, i, USK_i, \{ (j, UPK_j) \}_{j \in S}, PP,
	AU)$, adds $i$ to $DQ$, and responds $CK$ to $\mc{A}$.
    \end{itemize}

\item \tb{Challenge}: $\mc{A}$ submits a challenge set $S^* \subseteq KQ
\setminus CQ$. $\mc{C}$ obtains a ciphertext tuple $(CH^*, CK^*)$ by running
$\tb{Encaps} (S^*, \{ (j, UPK_j) \}_{j \in S^*}, PP, -)$. It sets $CK_0^* =
CK^*$ and $CK_1^* = RK$ by selecting a random $RK$. It flips a random coin
$\mu \in \bits$ and gives $(CH^*, CK_\mu^*)$ to $\mc{A}$.

\item \tb{Query Phase 2}: $\mc{A}$ continues to request decryption queries.
These queries are processed as follows:
	\begin{itemize}
	\item Decryption: $\mc{A}$ issues this query on $(S, CH, i, AU)$ such that
	$S \subseteq KQ \setminus CQ$, $i \in S$, and $CH \neq CH^*$. $\mc{C}$
	obtains $CK$ by running $\tb{Decaps} (S, CH, i, USK_i, \{ (j, UPK_j) \}_{j
	\in S}, PP, AU)$, adds $i$ to $DQ$, and responds $CK$ to $\mc{A}$.
	\end{itemize}

\item \tb{Guess}: Finally, $\mc{A}$ outputs a guess $\mu' \in \bits$, and
wins the game if $\mu = \mu'$.
\end{enumerate}
The advantage of $\mc{A}$ is defined as $\Adv_{DBE,\mc{A}}^{AD\text{-}CCA}
(\lambda) = \big| \Pr[\mu = \mu'] - \frac{1}{2} \big|$ where the probability
is taken over all the randomness of the experiment. A DBE scheme is AD-CCA
secure if for all probabilistic polynomial-time (PPT) adversary $\mc{A}$, the
advantage of $\mc{A}$ is negligible in the security parameter $\lambda$.
\end{definition}


The active-adaptive CCA security model is the extension of the adaptive CCA
security model that allows an attacker to register malicious public keys
\cite{KolonelosMW23}. In the active-adaptive security model, the index
belonging to the set of maliciously registered public keys is not included in
the challenge set, so it is generally known that an adaptive CPA secure DBE
scheme also satisfies active-adaptive CPA security. For the same reason, the
adaptive CCA secure DBE scheme is also active-adaptive CCA secure.

\begin{definition}[Active-Adaptive CCA Security] \label{def:dbe-aacca-sec}
The active-adaptive CCA (AA-CCA) security of DBE is defined in terms of the
following experiment between a challenger $\mc{C}$ and a PPT adversary
$\mc{A}$ where $1^\lambda$ and $1^L$ are given as input:
\begin{enumerate}
\item \tb{Setup}: $\mc{C}$ obtains public parameters $PP$ by running
$\tb{Setup}(1^{\lambda}, 1^L)$ and gives $PP$ to $\mc{A}$. It prepares $KQ,
CQ, MQ$, and $DQ$ as empty set.

\item \tb{Query Phase 1}: $\mc{A}$ adaptively requests key generation, key
corruption, malicious corruption, and decryption queries. These queries are
processed as follows:
    \begin{itemize}
    \item Key Generation: $\mc{A}$ issues this query on an index $i \in [L]$
	such that $i \not\in KQ \wedge i \not\in MQ$. $\mc{C}$ creates $(USK_i,
	UPK_i)$ by running $\tb{GenKey}(i, PP)$, adds $i$ to $KQ$, and responds
	$UPK_i$ to $\mc{A}$.

    \item Key Corruption: $\mc{A}$ issues this query on an index $i \in [L]$
	such that $i \in KQ \wedge i \not\in CQ$. $\mc{C}$ adds $i$ to $CQ$ and
	responds with $USK_i$ to $\mc{A}$.

	\item Malicious Corruption: $\mc{A}$ issues this query on an index $i \in
	[L]$ such that $i \not\in KQ \wedge i \not\in MQ$. $\mc{C}$ adds $i$ to
	$MQ$ and stores $UPK_i$.

	\item Decryption: $\mc{A}$ issues this query on $(S, CH, i, AU)$ such that
	$S \subseteq KQ \setminus ( CQ \cup MQ )$ and $i \in S$. $\mc{C}$ responds
	with $\tb{Decaps} (S, CH, i, USK_i, \{ (j, UPK_j) \}_{j \in S}, PP, AU)$.
    \end{itemize}

\item \tb{Challenge}: $\mc{A}$ submits a challenge set $S^* \subseteq KQ
\setminus (CQ \cup MQ)$. $\mc{C}$ obtains a ciphertext tuple $(CH^*, CK^*)$
by running $\tb{Encaps} (S^*, \{ (j, UPK_j) \}_{j \in S^*}, PP, -)$. It sets
$CK_0^* = CK^*$ and $CK_1^* = RK$ by selecting a random $RK$. It flips a
random coin $\mu \in \bits$ and gives $(CH^*, CK_\mu^*)$ to $\mc{A}$.

\item \tb{Query Phase 2}: $\mc{A}$ continues to request decryption queries.
These queries are processed as follows:
	\begin{itemize}
	\item Decryption: $\mc{A}$ issues this query on $(S, CH, i, AU)$ such that
	$S \subseteq KQ \setminus (CQ \cup MQ)$, $i \in S$, and $CH \neq CH^*$.
	$\mc{C}$ responds with $\tb{Decaps} (S, CH, i, USK_i, \{ (j, UPK_j) \}_{j
	\in S}, PP, AU)$.
	\end{itemize}

\item \tb{Guess}: Finally, $\mc{A}$ outputs a guess $\mu' \in \bits$, and wins
the game if $\mu = \mu'$.
\end{enumerate}
The advantage of $\mc{A}$ is defined as $\Adv_{DBE,\mc{A}}^{AA\text{-}CCA}
(\lambda) = \big| \Pr[\mu = \mu'] - \frac{1}{2} \big|$ where the probability
is taken over all the randomness of the experiment. A DBE scheme is AA-CCA
secure if for all probabilistic polynomial-time (PPT) adversary $\mc{A}$, the
advantage of $\mc{A}$ is negligible in the security parameter $\lambda$.
\end{definition}

\begin{lemma}[\cite{KolonelosMW23}] \label{lem:conv-ad-to-aa}
Let $\Pi_{AD}$ be an adaptively secure DBE scheme. Then $\Pi_{AD}$ is also
active-adaptively secure.
\end{lemma}

\section{DBE Construction}

In this section, we propose CCA secure efficient DBE schemes in bilinear
groups.

\subsection{SS-CCA Construction} \label{sec:dbe-sscca-scheme}

We first construct a semi-static CCA-secure DBE scheme. To design this DBE
scheme, we start from the KMW-DBE1 scheme of Kolonelos et al.
\cite{KolonelosMW23}, which is a DBE scheme modified from the BGW-BE scheme
of Boneh et al. \cite{BonehGW05}. The KMW-DBE1 scheme provides CPA security
with constant-size ciphertexts, constant-size private keys, linear-size
public keys, and linear-size public parameters. One way to design a
CCA-secure DBE scheme is to use the OTS scheme to provide ciphertext
integrity and bind the OTS public key to the ciphertext. However, this method
has a disadvantage in that the ciphertext size increases because the OTS
public key and OTS signature are included in the ciphertext. To reduce the
ciphertext size of the DBE scheme, we apply the BMW technique of Boyen et al.
\cite{BoyenMW05}, which provides ciphertext integrity by directly using the
element included in the ciphertext header. However, the BMW technique is not
simply applied to the DBE scheme that supports the distributed private key
generation. To solve this problem, we modify the BMW technique to provide CCA
even for DBE schemes where individual users have public keys by adding
additional public parameters.

In addition, the existing KMW-DBE1 scheme has a disadvantage that it requires
$L$ pairing operations to check whether the user public key is correct where
$L$ is the maximum number of users. One way to improve the user public key
verification is to use a batch verification method \cite{CamenischHP12}.
However, it is difficult to apply efficient batch verification because the
group membership check in $\G_2$ is inefficient when the public key contains
many $\G_2$ group elements. To solve this problem, we change the ciphertext
structure so that the ciphertext consists of $\G_2$ and $\G_1$ group elements
instead of $\G_1$ group elements. Because of this change, we can set most of
the user public key elements to be $\G_1$ group elements. In this case,
public key verification using batch verification is very efficient because
two pairing operations are needed for batch verification and the group
membership check in $\G_1$ is very efficient. The detailed description of our
DBE scheme for the semi-static CCA security is as follows:

\begin{description}
\item [$\tb{DBE}_{SS}.\tb{Setup}(1^\lambda, 1^L)$:] Let $\lambda$ be a
security parameter and $L$ be the maximum number of users.
It first generates asymmetric bilinear groups $\G, \hat{\G}, \G_T$ of prime
order $p$ with random generators $g, \hat{g}$ of $\G, \hat{\G}$ respectively.
It selects a random exponent $\alpha \in \Z_p$ and sets $\{ A_k = g^{\alpha^k}
\}_{k=1}^{2L+2}$, $\{ \hat{A}_k = \hat{g}^{\alpha^k} \}_{k=1}^{L+1}$. It also
selects a random exponent $\beta \in \Z_p$ and sets $B = g^\beta, \{ B_k =
A_k^{\beta} \}_{2 \leq k \leq L+1}$. Next, it chooses two hash functions
$\tb{H}_1$ and $\tb{H}_2$ such that $\tb{H}_1 : \hat{\G} \rightarrow \Z_p$ and
$\tb{H}_2 : \G_T \rightarrow \bits^\lambda$.
It outputs public parameters
	\begin{align*}
    PP = \Big(
    &   (p, \G, \hat{\G}, \G_T, e),~ g,~ \hat{g},~
		\big\{ A_k \big\}_{ 1 \leq k \neq L+2 \leq 2L+2},~
		\big\{ \hat{A}_k \}_{1 \leq k \leq L+1},~
		B,~ \big\{ B_k \big\}_{2 \leq k \leq L+1},~ \\
	&	\Omega = e(A_{L+2}, \hat{g}),~ \tb{H}_1,~ \tb{H}_2
    \Big).
    \end{align*}

\item [$\tb{DBE}_{SS}.\tb{GenKey}(i, PP)$:] Let $i \in [L]$. It selects a
random exponent $\gamma_i \in \Z_p$ and builds a private key $USK_i$ and a
public key $UPK_i$ as
	\begin{align*}
    USK_i = \Big(
		K_i = A_{L+2-i}^{\gamma_i} \Big),~
    UPK_i = \Big(
        V_i = g^{\gamma_i},~ \hat{V}_i = \hat{g}^{\gamma_i},~
        \big\{
		V_{i,k} = A_k^{\gamma_i}
		\big\}_{2 \leq k \neq L+2-i \leq L+1}
    	\Big).
    \end{align*}

\item [$\tb{DBE}_{SS}.\tb{IsValid}(j, UPK_j, PP)$:] Let $UPK_j = (V_j,
\hat{V}_j, \{ V_{j,k} \})$.
It first checks that $UPK_j = (V_j, \hat{V}_j, \{ V_{j,k} \}) \in \G \times
\hat{\G} \times \G^{2L}$.
It chooses small random exponents $\delta_0, \{ \delta_k \}_{k=2}^{2L+1}$ where
$\delta_0, \delta_k$ are elements of $\ell_b$ bit from $\Z_p$ and checks that
	\begin{align*}
	e\bigg( V_j^{\delta_0} \prod_{2 \leq k \neq L+2-j \leq 2L+2}
		V_{j,k}^{\delta_k}, \hat{g} \bigg) \stackrel{?}{=}
	e\bigg( g^{\delta_0} \prod_{2 \leq k \neq L+2-j \leq 2L+2} A_k^{\delta_k},
		\hat{V}_j \bigg).
	\end{align*}
    If the equation holds, then it outputs 1. Otherwise, it outputs 0.

\item [$\tb{DBE}_{SS}.\tb{Encaps}(S, \{ (j, UPK_j) \}_{j \in S}, PP, AU)$:]
Let $UPK_j = (V_j, \hat{V}_j, \{ V_{j,k} \})$.
It selects a random exponent $t \in \Z_p$ and computes a tag $\omega = \tb{H}_1
(\hat{g}^t \| AU)$. It outputs a ciphertext header
    \begin{align*}
    CH = \Big(
    	\hat{C}_1 = \hat{g}^{t},~
    	C_2 = \big( A_{L+1}^\omega B \prod_{j \in S} A_j V_j \big)^t
    \Big)
    \end{align*}
and a session key $CK = \tb{H}_2(\Omega^t)$.

\item [$\tb{DBE}_{SS}.\tb{Decaps}(S, CH, i, USK_i, \{ (j, UPK_j) \}_{j \in S},
PP, AU)$:] Let $CH = (\hat{C}_1, C_2)$, $USK_i = K_i$, and $UPK_j = (V_j,
\hat{V}_j, \lb \{ V_{j,k} \})$. If $i \not\in S$, it outputs $\perp$.
It computes a tag $\omega = \tb{H}_1(\hat{C}_1 \| AU)$ and checks the validity
of the ciphertext
	\begin{align*}
	e\big( C_2, \hat{g} \big) \stackrel{?}{=}
	e\Big( A_{L+1}^\omega B \prod_{j \in S} A_j V_j, \hat{C}_1 \Big).
	\end{align*}
If the equation fails, then it outputs $\perp$.
It selects a random exponent $r \in \Z_p$ and builds decryption components
	\begin{align*}
	D_1 &= K_i \cdot \Big( A_{L+1}^\omega B \prod_{j \in S} A_j V_j \Big)^r,~
	\hat{D}_2 = \hat{A}_{L+2-i} \cdot \hat{g}^r,~ \\
	D_3 &= A_{2L+3-i}^\omega B_{L+2-i},~
	D_4  = \prod_{j \in S \setminus \{i\}} A_{L+2-i+j} V_{j,L+2-i}.
	\end{align*}
It outputs a session key
   	$CK = \tb{H}_2( e(C_2, \hat{D}_2) \cdot e(D_1 \cdot D_3 \cdot D_4,
		\hat{C}_1)^{-1} )$.
\end{description}

\subsection{AD-CCA Construction} \label{sec:dbe-adcca-scheme}

Now, we construct an adaptive CCA-secure DBE scheme from the semi-static
CCA-secure DBE scheme. The GW transformation of Gentry and Waters
\cite{GentryW09} is used to convert a semi-static CPA-secure BE or DBE scheme
into the adaptive CPA-secure BE or DBE scheme. However, the GW transformation
is applied to the CPA-secure BE or DBE scheme, and it needs additional
modification to be applied to the CCA-secure BE or DBE scheme. We derive the
adaptive CCA-secure DBE scheme by combining the semi-static CCA-secure DBE
scheme, the OTS scheme, and the GW transformation. Here, we use the OTS
signature for ciphertext integrity, and we modify the encryption and
decryption algorithms of the underlying DBE scheme to additionally receive a
label as input to bind the OTS public key with the DBE ciphertext. This
modification of algorithms to receive an additional label as input was widely
used in the design of existing CCA-secure PKE schemes \cite{Shoup01,Kiltz06}.
The detailed description of our DBE scheme for adaptive CCA security is given
as follows:

\begin{description}
\item [$\tb{DBE}_{AD}.\tb{Setup}(1^\lambda, 1^L)$:] Let $\lambda$
be a security parameter and $L$ be the number of users. It obtains $PP_{SS}$
by running $\tb{DBE}_{SS}.\tb{Setup}(1^\lambda, 1^{2L})$. It outputs public
parameters $PP = PP_{SS}$.

\item [$\tb{DBE}_{AD}.\tb{GenKey}(i, PP)$:] Let $i \in [L]$.  It generates two
key pairs $(USK_{SS,2i}, UPK_{SS,2i})$ and $(USK_{SS,2i-1}, \lb
UPK_{SS,2i-1})$ by running $\tb{DBE}_{SS}.\tb{GenKey}(2i, PP)$ and
$\tb{DBE}_{SS}.\tb{GenKey}(2i-1, PP)$ respectively. It selects a random bit
$u_i \in \bits$ and erases $USK_{SS,2i-(1-u_i)}$ completely. It outputs a
private key $USK_i = (USK_{SS,2i-u_i}, u_i)$ and a public key $UPK_i =
(UPK_{SS,2i}, UPK_{SS,2i-1})$.

\item [$\tb{DBE}_{AD}.\tb{IsValid}(j, UPK_j, PP)$:] Let $UPK_j =
(UPK_{SS,2j}, UPK_{SS,2j-1})$. It checks that $\tb{DBE}_{SS}.\tb{IsValid}(2j,
\lb UPK_{SS,2j}, PP_{SS}) = 1$ and $\tb{DBE}_{SS}.\tb{IsValid}(2j-1,
UPK_{SS,2j-1}, PP_{SS}) = 1$. If it passes all checks, then it outputs 1.
Otherwise, it outputs 0.

\item [$\tb{DBE}_{AD}.\tb{Encaps}(S, \{ (j, UPK_j) \}_{j \in S},
PP, AU)$:] Let $S \subseteq [L]$ and $UPK_j = (UPK_{SS,2j}, UPK_{SS,2j-1})$.
	\begin{enumerate}
	\item It generates $(SK, VK)$ by running $\tb{OTS.GenKey} (1^\lambda)$.

	\item It selects random bits $z = \{ z_j \}_{j \in S}$ where $z_j \in \bits$.
	Next, it defines
	two sets $S_0 = \{ 2j - z_j \}_{j \in S}$ and $S_1 = \{ 2j - (1-z_j) \}_{j
	\in S}$.

	\item It obtains ciphertext pairs $(CH_{SS,0}, CK_{SS,0})$ and
	$(CH_{SS,1}, CK_{SS,1})$ by running $\tb{DBE}_{SS}.\tb{Encaps} (S_0, \lb
	\{ (k, UPK_{SS,k} \}_{k \in S_0}, PP_{SS}, VK)$ and $\tb{DBE}_{SS}.
	\tb{Encaps} (S_1, \{ (k, UPK_{SS,k} \}_{k \in S_1}, PP_{SS}, VK)$
	respectively.

	\item It selects a random message $CK \in \bits^\lambda$. It obtains
	ciphertexts $CT_{0}$ and $CT_{1}$ by running $\tb{SKE.Encrypt} (CK_{SS,0},
	CK)$ and $\tb{SKE.Encrypt} (CK_{SS,1}, CK)$ respectively.

	\item It sets $CM = (CH_{SS,0}, CH_{SS,1}, CT_{0}, CT_{1}, z)$ and
	calculates $\sigma$ by running $\tb{OTS.Sign} (SK, CM)$.

	\item It outputs a ciphertext header $CH = (CM, \sigma, VK)$ and a session
	key $CK$.
	\end{enumerate}

\item [$\tb{DBE}_{AD}.\tb{Decaps}(S, CH, i, USK_i, \{ (j, UPK_j) \}_{j \in S},
PP, AU)$:] Let $CH = (CM, \sigma, VK)$ and $USK_i =  (USK_{SS,2i-u_i}, \lb u_i)$
where $CM = (CH_{SS,0}, CH_{SS,1}, CT_{0}, CT_{1}, z)$.
If $i \not\in S$, it outputs $\perp$.
	\begin{enumerate}
	\item It checks that $1 \stackrel{?}{=} \tb{OTS.Verify} (VK, \sigma, CM)$.
	If this check fails, it outputs $\perp$.

	\item It derives two sets $S_0 = \{ 2j - z_j \}_{j \in S}$ and $S_1 = \{
	2j - (1-z_j) \}_{j \in S}$. If $z_i = u_i$, then it sets $S' = S_0, CH'_{SS}
	= CH_{SS,0}, CT' = CT_{0}$. Otherwise, it sets $S' = S_1, CH'_{SS} =
	CH_{SS,1}, CT' = CT_{1}$.
	
	\item It obtains $CK'_{SS}$ by running $\tb{DBE}_{SS}.\tb{Decaps} (S',
	CH'_{SS}, 2i-u_i, USK_{SS,2i-u_i}, \{ (k, UPK_{SS,k} \}_{k \in S'}, PP_{SS},
	\lb VK)$.

	\item It obtains $CK$ by running $\tb{SKE.Decrypt} (CK'_{SS}, CT')$ and
	outputs a session key $CK$.
	\end{enumerate}
\end{description}

\subsection{Correctness}

\begin{theorem} \label{thm:dbe-sscca-correct}
The above DBE$_{SS}$ scheme is correct.
\end{theorem}

\begin{proof}
To show the correctness of this scheme, we first show that the same session
key can be derived. If $i \in S$, then we can derive the following equation
 	\begin{align*}
    &   e(C_2, \hat{A}_{L+2-i}) \\
    &=  e( \big( A_{L+1}^\omega B \prod_{j \in S} A_j V_j \big)^t,
		   \hat{A}_{L+2-i}) \\
    &=  e( \big( A_{L+1}^\omega B \cdot A_i V_i \cdot
		   \prod_{j \in S \setminus \{i\}} A_j V_j \big)^t, \hat{A}_{L+2-i} ) \\
    &=  e( (A_{L+1}^{\omega} B)^t, \hat{A}_{L+2-i} ) \cdot
		e( A_i^t, \hat{A}_{L+2-i} ) \cdot e( V_i^t, \hat{A}_{L+2-i} ) \cdot
    	e( \big( \prod_{j \in S \setminus \{i\}} A_j V_j \big)^t, \hat{A}_{L+2-i} ) \\
    &=  e( A_{2L+3-i}^\omega B_{L+2-i}, \hat{g}^t ) \cdot
		e( A_{L+2}, \hat{g}^t ) \cdot
		e( \hat{A}_{L+2-i}^{\gamma_i}, \hat{g}^t ) \cdot
    	e( \big( \prod_{j \in S \setminus \{i\}} A_{L+2-i+j} A_{L+2-i}^{\gamma_j},
		   \hat{g}^t \big) \\
    &=  e( D_3, \hat{g}^t ) \cdot \Omega^t \cdot e( K_i, \hat{g}^t, ) \cdot
		e( D_4, \hat{g}^t ) \\
    &=  \Omega^t \cdot e( K_i \cdot D_3 \cdot D_4, \hat{C}_1 ).
    \end{align*}
By using the above equation, we can obtain that the same element that is used
to derive a session key as follows
    \begin{align*}
    &   e(C_2, \hat{D}_2) \cdot e(D_1 \cdot D_3 \cdot D_4, \hat{C}_1)^{-1} \\
    &=  e(C_2, A_{L+2-i}) \cdot e(C_2, \hat{g}^r) \cdot
		e(D_1 \cdot D_3 \cdot D_4, \hat{C}_1)^{-1} \\
    &=  \Omega^t \cdot e(K_i \cdot D_3 \cdot D_4, \hat{C}_1) \cdot
		e( \big( A_{L+1}^\omega B \prod_{j \in S} A_j V_j \big)^t, \hat{g}^r ) \cdot
		e(D_1^{-1}, \hat{C}_1) \cdot e((D_3 \cdot D_4)^{-1}, \hat{C}_1) \\
    &=  \Omega^t \cdot e(K_i \cdot \big( A_{L+1}^\omega B \prod_{j \in S}
			A_j V_j \big)^r, \hat{C}_1) \cdot e(D_1^{-1}, \hat{C}_1)
     =  \Omega^t.
    \end{align*}

Next, we show that the verification of a public key is correct. Since our
public key verification uses the small exponent test, the validity of this
batch verification can be checked by the previous work of Camenisch et al.
\cite{CamenischHP12}, we omit the detailed analysis of this batch
verification.
\end{proof}

\begin{theorem} \label{thm:dbe-adcca-correct}
The above DBE$_{AD}$ scheme is correct.
\end{theorem}

The correctness of the DBE$_{AD}$ scheme can be easily derived from the
correctness of the DBE$_{SS}$, SKE, and OTS schemes. We omit the proof of this
theorem.

%

\section{Security Analysis}

In this section, we prove the semi-static CCA security and adaptive CCA
security of the proposed DBE schemes.

\subsection{Semi-Static Security Analysis}

The basic idea of the semi-static CCA proof is to use the partitioning
technique that divides the ciphertext header space into two regions: a
challenge ciphertext header region and a decryption oracle ciphertext header
region. The challenge ciphertext header $CH^*$ is associated with a tag
$\omega^*$, and the decryption oracle ciphertext $CH$ is associated with a
tag $\omega$. If $\omega^* \neq \omega$, a simulator first derives a
decryption key to be used for decrypting the ciphertext header $CH$ and
perform the decryption of $CH$ to handle the decryption query. In other
words, the simulator can decrypt the ciphertext header requested by the
attacker by using a method like the IBE private key derivation for an
identity $\omega$. If $\omega^* = \omega$ holds, the simulator can find the
collisions of the collision-resistant hash function, which ensures that such
an event will not occur. The detailed description of the semi-static CCA
proof is given as follows:

\begin{theorem} \label{thm:dbe-sscca-sec}
The above DBE$_{SS}$ scheme is SS-CCA secure if the $(L+1)$-BDHE assumption
holds and $\tb{H}_1$ is a collision-resistant hash function.
\end{theorem}

\begin{proof}
Suppose there exists an adversary $\mc{A}$ that breaks the DBE$_{SS}$ scheme
with a non-negligible advantage. A simulator $\mc{B}_1$ that solves the
$L+1$-BDHE assumption using $\mc{A}$ is given: a challenge tuple $D = (g, g^a,
\ldots, g^{a^{L+1}}, g^{a^{L+3}}, \ldots, g^{2L+2}, g^b, \hat{g}, \hat{g}^a,
\ldots, \hat{g}^{a^{L+1}}, \hat{g}^b)$ and $Z$ where $Z = Z_0 =
e(g,\hat{g})^{a^{L+2} b}$ or $Z = Z_1 = e(g,\hat{g})^c$. Then $\mc{B}_1$ that
interacts with $\mc{A}$ is described as follows:

\svs\noindent \tb{Init:} $\mc{A}$ commits an initial set $\tilde{S}$.

\svs\noindent \tb{Setup:} $\mc{B}_1$ sets $\hat{C}_1^* = \hat{g}^b$ and computes
$\omega^* = \tb{H}_1(\hat{C}_1^*)$. It selects a random exponent $\beta'$ and
implicitly sets $\beta = -\omega^* a^{L+1} + \beta'$. It creates public
parameters
	\begin{align*}
	& g,~ \hat{g},~
	\big\{ A_k = g^{a^k} \big\}_{1 \leq k \neq L+2 \leq 2L+2},~
	\big\{ \hat{A}_k = \hat{g}^{a^k} \big\}_{1 \leq k \leq L+1},~ \\
	& B = A_{L+1}^{-\omega^*} g^{\beta'},~
	\big\{ B_k = A_{L+1+k}^{-\omega^*} A_k^{\beta'} \big\}_{2 \leq k \leq L+1},~
	\Omega = e(A_{L+1}, \hat{A}_1).
	\end{align*}

\svs\noindent \tb{Query Phase 1:} For each index $i \in \tilde{S}$, it selects
random $\gamma'_i \in \Z_N$ and creates a public key by implicitly setting
$\gamma_i = \gamma'_i - \alpha^i$ as
	\begin{align*}
	UPK_i = \big(
		V_j = A_j^{-1} g^{\gamma'_j},~
		\hat{V}_j = \hat{A}_j^{-1} \hat{g}^{\gamma'_j},~
		\big\{ V_{j,k} = A_{k+j}^{-1} A_{k}^{\gamma'_j}
		\big\}_{2 \leq k \neq L+2-j \leq L+1}
	\big).
	\end{align*}
It gives $\{ (j, UPK_j) \}_{j \in \tilde{S}}$ to $\mc{A}$.

\svs For a decryption query on $(S, CH = (\hat{C}_1, C_2), i, AU)$ such that
$S \subseteq \tilde{S}$, $\mc{B}_1$ proceeds as follows:
It first computes $\omega = \tb{H}_1(\hat{C}_1 \| AU)$ and checks the validity of
the ciphertext header by using pairing. If the check fails, then it responds
with $\perp$. If $\omega = \omega^*$, then it sets $\ts{Bad} = 1$ and aborts
the simulation.
Otherwise ($\omega \neq \omega^*$), it selects a random exponent $r' \in \Z_p$
and derives decryption components by implicitly setting $r = a/(\omega -
\omega^*) + r'$ as
	\begin{align*}
	D_1 &= A_{L+2-i}^{\gamma'_i} (A_1^{\beta'} \prod_{j \in S}
		A_{j+1} V_{j,1})^{ 1/(\omega - \omega^*) }
	 	\big( A_{L+1}^{\omega} B \prod_{j \in S} A_j V_j \big)^{r'},~ \\
	D_2 &= A_{L+2-i} A_1^{1/(\omega - \omega^*)} g^{r'},~ \\
	D_3 &= A_{2L+3-i}^\omega B_{L+2-i},~
	D_4 = \prod_{j \in S \setminus \{i\}} A_{L+2-i+j} V_{j,L+2-i}.
	\end{align*}
It responds a session key $CK = \tb{H}_2( e(C_2, D_2) \cdot e(D_1 \cdot D_3
\cdot D_4, \hat{C}_1)^{-1} )$ to $\mc{A}$.

\vs \noindent \tb{Challenge}: $\mc{A}$ submits a challenge set $S^* \subseteq
\tilde{S}$. $\mc{B}_1$ implicitly sets $t = b$ and creates a ciphertext tuple
	\begin{align*}
	CH^* = \big(
		\hat{C}_1^* = \hat{g}^b,~
		C_2^* = \big( g^b \big)^{\beta' + \sum_{j \in S^*} \gamma'_j}
	\big),~
	CK^* = \tb{H}_2(Z).
	\end{align*}
It sets $CK_0^* = CK^*$ and $CK_1^* = \tb{H}_2(R)$ by selecting a random
element $R \in \G_T$. It flips a random bit $\mu \in \bits$ and gives $(CH^*,
CK_\mu^*)$ to $\mc{A}$.

\svs \noindent \tb{Query Phase 2}: The decryption queries are handled as the
same as query phase 1.

\svs \noindent \tb{Guess}: Finally, $\mc{A}$ outputs a guess $\mu' \in
\bits$. $\mc{B}_1$ outputs $0$ if $\mu = \mu'$ or $1$ otherwise.

The decryption components are correctly distributed as
	\begin{align*}
	D_1
	&=	K_i \big( A_{L+1}^\omega B \prod_{j \in S} A_j V_j \big)^r
	 = 	A_{L+2-i}^{\gamma'_i - \alpha^i}
	 	\big( A_{L+1}^{\omega - \omega^*} g^{\beta'}
		\prod_{j \in S} A_j V_j \big)^{a/(\omega - \omega^*) + r'} \\
	&= 	A_{L+2-i}^{\gamma'_i} A_{L+2}^{-1}
		A_{L+2} \big( g^{\beta'} \prod_{j \in S} A_j V_j \big)^{
			a/(\omega - \omega^*) }
		\big( A_{L+1}^{\omega} B \prod_{j \in S} A_j V_j \big)^{r'} \\
	&=	A_{L+2-i}^{\gamma'_i} (A_1^{\beta'} \prod_{j \in S} A_{j+1} V_{j,1})^{
			1/(\omega - \omega^*) }
	 	\big( A_{L+1}^{\omega} B \prod_{j \in S} A_j V_j \big)^{r'},~ \\
	D_2
	&= 	A_{L+2-i} g^r = A_{L+2-i} g^{a/(\omega - \omega^*) + r'}
	 =	A_{L+2-i} A_1^{1/(\omega - \omega^*)} g^{r'}.
	\end{align*}
The challenge ciphertext header are also correctly distributed as
	\begin{align*}
	C_2^*
	&=	\big( A_{L+1}^{\omega^*} B \prod_{j \in S^*} A_j V_j \big)^t
	 =	\big( A_{L+1}^{\omega^*} A_{L+1}^{-\omega^*} g^{\beta'}
	 	\prod_{j \in S^*} g^{\alpha^j} g^{\gamma'_j - \alpha^j} \big)^t
	 =	\big( g^t \big)^{\beta' + \sum_{j \in S^*} \gamma'_j}.
	\end{align*}

Suppose there exists an adversary $\mc{A}$ that breaks the DBE$_{SS}$ scheme
with a non-negligible advantage. A simulator $\mc{B}_2$ that breaks the
CRHF using $\mc{A}$ is given the description of $\tb{H}_1$.
Then $\mc{B}_2$ that interacts with $\mc{A}$ is described as follows:

\svs\noindent \tb{Init:} $\mc{A}$ commits an initial set $\tilde{S}$.

\svs\noindent \tb{Setup:} $\mc{B}_2$ creates public parameters by running the
setup algorithm except that $\tb{H}_1$ is replaced by the given hash function. It
selects a random exponent $t \in \Z_p$, sets $\hat{C}_1^* = \hat{g}^t$, and
computes $\omega^* = \tb{H}_1(\hat{C}_1^*)$.

\svs\noindent \tb{Query Phase 1:} For all $j \in \tilde{S}$, $\mc{B}_2$ generates
$(USK_j, UPK_j)$ by simply running the key generation algorithm. It gives
$\{ (j, UPK_j) \}_{j \in \tilde{S}}$ to $\mc{A}$.
For a decryption query on $(S, CH = (\hat{C}_1, C_2), i, AU)$, it first
computes $\omega = \tb{H}_1(\hat{C}_1 \| AU)$. If $\omega = \omega^*$, then it sets
$\ts{Bad} = 1$ and outputs a collision pair $(\hat{C}_1^*, \hat{C}_1 \| AU)$.
Otherwise, it responds with a session key by simply running the decryption
algorithm.

\svs\noindent \tb{Challenge:} To create a ciphertext $(CH^*, CK^*)$,
$\mc{B}_2$ simply runs the encapsulation algorithm except that it uses the
exponent $t$ chosen in the setup. It sets $CK_0^* = CK^*$ and $CK_1^* = RK$
by selecting a random $RK$. It flips a coin $\mu \in \bits$ and gives $(CH^*,
CK_\mu^*)$ to $\mc{A}$.

\svs \noindent \tb{Query Phase 2}: The decryption queries are handled as the
same as query phase 1.

\svs\noindent \tb{Guess:} $\mc{B}_2$ outputs $\perp$.

By combining the results of the above simulations, we obtain the following
equation
	\begin{align*}
	\Adv_{DBE}^{SS\text{-}CCA}(\lambda)
	&= \Pr[\mu = \mu'] - 1/2 \\
	&= \Pr[\neg \ts{Bad}] \cdot \Pr[\mu = \mu' | \neg \ts{Bad}] +
			\Pr[\ts{Bad}] \cdot \Pr[\mu = \mu' | \ts{Bad}] - 1/2 \\
	&\leq \Pr[\mu = \mu' | \neg \ts{Bad}] - 1/2 + \Pr[\ts{Bad}]
	 \leq \Adv_{B_1}^{BDHE}(\lambda) + \Adv_{B_2}^{CRHF}(\lambda).
	\end{align*}
This completes our proof.
\end{proof}

\subsection{Adaptive Security Analysis}

The basic idea of the adaptive CCA proof is to change the challenge session
key to a random value through hybrid games so that the attacker can never
distinguish the coin thrown by the challenger. To do this, we first play
hybrid games to change the secret keys used for $CT_0^*, CT_1^*$ in the
challenge ciphertext header $CH^*$ to random values by using the semi-static
CCA security of the underlying DBE$_{SS}$ scheme. Then, we play hybrid games
to change the messages of $CT_0^*, CT_1^*$ to random values by using the OMI
security of the underlying SKE scheme. In the last game, the challenge
session key $CK_\mu^*$ is not related to the challenge ciphertext header
$CH^*$. The detailed description of the adaptive CCA proof is given as
follows:

\begin{theorem}[Adaptive CCA Security]
The above DBE$_{AD}$ scheme is AD-CCA secure if the DBE$_{SS}$ scheme is SS-CCA
secure and the OTS scheme is strongly unforgeable.
\end{theorem}

\begin{proof}
The security proof consists of a sequence of hybrid games $\tb{G}_0, \tb{G}_1,
\ldots, \tb{G}_4$. The first game will be the original adaptive CCA security
game and the last one will be a game in which an adversary has no advantage.
We define the games as follows:
\begin{description}
\item [\tb{Game} $\tb{G}_0$.] This game is the original security game defined
in Section \ref{def:dbe-adcca-sec}. That is, the simulator of this game simply
follows the honest algorithms.

\item [\tb{Game} $\tb{G}_1$.] This game is the same as the game $\tb{G}_0$
    except that the ciphertext $CT_0$ in $CH^*$ is created as $CT_0 =
    \tb{SKE.Encrypt}(RK_{SS,0}, CK_0)$ by using a random $RK_{SS,0}$
    instead of a valid $CK_{SS,0}$.

\item [\tb{Game} $\tb{G}_2$.] This game is also similar to the game $\tb{G}_1$
except that the ciphertext $CT_1$ in $CH^*$ is created as $CT_1 =
\tb{SKE.Encrypt}(RK_{SS,1}, CK_0)$ by using random $RK_{SS,1}$ instead of a
valid $CK_{SS,1}$.

\item [\tb{Game} $\tb{G}_3$.] This game is the same as the game $\tb{G}_2$
except that the ciphertext $CT_0$ in $CH^*$ is generated as $CT_0 =
\tb{SKE.Encrypt}(RK_{SS,0}, RK_0)$ by selecting a random $RK_0$ instead of a
valid $CK$.

\item [\tb{Game} $\tb{G}_4$.] This final game $\tb{G}_4$ is also similar to
the game $\tb{G}_3$ except that the ciphertext $CT_1$ in $CH^*$ is
generated as $CT_1 = \tb{SKE.Encrypt}(RK_{SS,1}, RK_1)$ by selecting a
random $RK_1$ instead of a valid $CK$. In this final game, $CH^*$ is not
related to the session key $CK_\mu^*$. Thus, the advantage of $\mc{A}$ is
zero.
\end{description}
Let $\Adv_{\mc{A}}^{G_j}$ be the advantage of $\mc{A}$ in the game $\tb{G}_j$.
We have that $\Adv_{DBE,\mc{A}}^{AD\text{-}CCA}(\lambda) =
\Adv_{\mc{A}}^{G_0}$, and $\Adv_{\mc{A}}^{G_4} = 0$. From the following Lemmas
\ref{lem:dbe-adcca-g0-g1}, \ref{lem:dbe-adcca-g1-g2},
\ref{lem:dbe-adcca-g2-g3}, and \ref{lem:dbe-adcca-g3-g4}, we obtain the
equation
	\begin{align*}
    \Adv_{DBE}^{AD\text{-}CCA}(\lambda)
    &\leq \sum_{j=1}^4 \big| \Adv_{\mc{A}}^{G_{j-1}} - \Adv_{\mc{A}}^{G_j} \big|
	+ \Adv_{\mc{A}}^{G_4} \\
    &\leq 2 \Adv_{DBE}^{SS\text{-}CCA}(\lambda) +
	 2 \Adv_{OTS}^{SUF}(\lambda) +
	 2 \Adv_{SKE}^{OMI}(\lambda).
    \end{align*}
This completes the proof.
\end{proof}

\begin{lemma} \label{lem:dbe-adcca-g0-g1}
If the DBE$_{SS}$ scheme is SS-CCA secure and the OTS scheme is SUF secure,
then no PPT adversary can distinguish $\tb{G}_0$ from $\tb{G}_1$ with a
non-negligible advantage.
\end{lemma}

\begin{proof}
Suppose there exists an adversary $\mc{A}$ that distinguishes $\tb{G}_0$ from
$\tb{G}_1$ with a non-negligible advantage. Let $\mc{C}$ be the challenger of
the DBE$_{SS}$ scheme with an input $1^{2L}$. The simulator $\mc{B}_1$ that
breaks the DBE$_{SS}$ scheme by using $\mc{A}$ with an input $1^L$ is
described as follows:

\vs\noindent \tb{Setup:} In this step, $\mc{B}_1$ proceeds as follows:
	\begin{enumerate}
	\item It chooses random bits $z^* = \{ z_j^* \}_{j \in [L]}$ where $z_j
	\in \bits$ and defines an initial set $\tilde{S}_{SS} = \{ 2j - z_j^*
	\}_{j=1}^{L}$.
	It also generates $(SK^*, VK^*)$ by running $\tb{OTS.GenKey} (1^\lambda)$.

	\item It commits $\tilde{S}_{SS}$ to $\mc{C}$ and receives $PP_{SS}$.
	It sets $PP = PP_{SS}$ and gives $PP$ to $\mc{A}$.

	\item Next, it requests a key generation query to $\mc{C}$ and
	receives all public keys $\{ (k, UPK_{SS,k}) \}_{k \in \tilde{S}_{SS}}$.
	\end{enumerate}

\svs\noindent \tb{Query Phase 1:} For the key generation query of $\mc{A}$ on
an index $i \in [L]$ such that $i \notin KQ$, $\mc{B}_1$ proceeds as follows:
	\begin{enumerate}
	\item It sets $k = 2i - (1-z_i^*)$ and obtains $(USK_{SS,k}, UPK_{SS,k})$
	by running $\tb{DBE}_{SS}.\tb{GenKey} (k, PP_{SS})$. It stores $USK_{SS,k}$
	to $KL$.

	\item It sets $UPK_i = (UPK_{SS,2i}, UPK_{SS,2i-1})$, adds $i$ to $KQ$,
	and gives $UPK_i$ to $\mc{A}$.
	\end{enumerate}

\svs For the key corruption query of $\mc{A}$ on an index $i \in [L]$ such
that $i \in KQ \setminus DQ$, $\mc{B}_1$ proceeds as follows:
	\begin{enumerate}
	\item It fixes $u_i = 1-z_i^*$ and retrieves $USK_{SS,2i-u_i}$ from $KL$.

	\item It sets $USK_i = (USK_{SS, 2i-u_i}, u_i)$, adds $i$ to $CQ$, and
	gives $USK_i$ to $\mc{A}$.
	\end{enumerate}

\svs For the decryption query of $\mc{A}$ on $(S, CH = (CM, \sigma, VK), i,
AU)$ such that $S \subseteq KQ \setminus CQ$ where $CM = (CH_{SS,0}, CH_{SS,1},
\lb CT_{0}, CT_{1}, z)$, $\mc{B}_1$ proceeds as follows:
	\begin{enumerate}
	\item If $VK = VK^*$, then it sets $\ts{Bad} = 1$ and aborts the
	simulation.

	\item It first checks $1 \stackrel{?}{=} \tb{OTS.Verify} (CM, \sigma,
	VK)$. If this check fails, it gives $\perp$ to $\mc{A}$.

	\item If $(i, u_i) \in UL$, then it retrieves $u_i$. Otherwise, It selects
	a random bit $u_i \in \bits$ and adds $(i, u_i)$ to $UL$.

	\item It defines $S_0 = \{ 2j - z_j \}_{j \in S}$ and $S_1 = \{ 2j -
	(1-z_j) \}_{j \in S}$. If $z_i = u_i$, then it sets $S' = S_0, CH'_{SS} =
	CH_{SS,0}, CT' = CT_{0}$. Otherwise, it sets $S' = S_1, CH'_{SS} =
	CH_{SS,1}, CT' = CT_{1}$.

	\item If $u_i = 1-z_i^*$, then it retrieves $USK_{SS,2i-u_i}$ from $KL$
	and obtains $CK'_{SS}$ by running $\tb{DBE}_{SS}.\tb{Decaps} (S', \lb CH'_{SS},
	2i-u_i, USK_{SS,2i-u_i}, \{ (k, UPK_{SS,k}) \}_{k \in S'}, PP_{SS}, VK)$.
	Otherwise ($u_i = z_i^*$), it requests a decryption query on $(S',
	CH'_{SS}, 2i-u_i, VK)$ to $\mc{C}$ and receives $CK'_{SS}$.

	\item It derives $CK$ by running $\tb{SKE.Decrypt} (CT', CK'_{SS})$.
	It adds $i$ to $DQ$ and gives $CK$ to $\mc{A}$.
	\end{enumerate}

\svs\noindent \tb{Challenge:} $\mc{A}$ submits a challenge set $S^* \subseteq
[L]$ such that $S^* \subseteq KQ \setminus CQ$, $\mc{B}_1$ proceeds as follows:
	\begin{enumerate}
	\item It defines $S_0^* = \{ 2j - z_j^* \}_{j \in S^*}$ and $S_1^* =
	\{ 2j - (1-z_j^*) \}_{j \in S^*}$.
	It submits a challenge set $S_{0}^*$ to $\mc{C}$ and receives
	$(CH_{SS,0}^*, CK_{SS,0}^*)$. It obtains $(CH_{SS,1}^*, CK_{SS,1}^*)$
	by running $\tb{DBE}_{SS}.\tb{Encaps}(S_1^*, \{ (k, UPK_{SS,k}) \}_{k \in
	S_1^*}, \lb PP_{SS}, VK^*)$.

	\item It selects a session key $CK^*$. It obtains $CT_0^*$ and $CT_1^*$
	by running $\tb{SKE.Encrypt}(CK_{SS,0}^*, CK^*)$ and $\tb{SKE.Encrypt}
	(CK_{SS,1}^*, CK^*)$ respectively.

	\item It sets $CM^* = (CH_{SS,0}^*, CH_{SS,1}^*, CT_0^*, CT_1^*, z^*)$
	and calculates $\sigma^*$ by running $\tb{OTS.Sign} (SK^*, CM^*)$.

	\item It sets $CH^* = (CM^*, \sigma^*, VK^*)$, $CK_0^* = CK^*$, and
	$CK_1^* = RK$ by selecting a random $RK$. Next, it flips a random coin
	$\mu \in \bits$ and gives $(CH^*, CK_\mu^*)$ to $\mc{A}$
	\end{enumerate}

\svs\noindent \tb{Query Phase 2:} For a decryption query on $(S, CH = (CM,
\sigma, VK), i, AU)$ such that $S \subseteq S^* \wedge CH \neq CH^*$,
$\mc{B}_1$ handles this decryption query as the same as query phase 1.

\svs\noindent \tb{Guess:} $\mc{A}$ outputs a guess $\mu'$. If $\mu = \mu'$,
then $\mc{B}_1$ outputs 1. Otherwise, it outputs 0.

Suppose there exists an adversary $\mc{A}$ that breaks the DBE$_{SS}$ scheme
with a non-negligible advantage. A simulator $\mc{B}_2$ that forges the
OTS scheme using $\mc{A}$ is given $VK^*$.
Then $\mc{B}_2$ that interacts with $\mc{A}$ is described as follows:

\svs\noindent \tb{Setup:} $\mc{B}_2$ creates public parameters by running the
setup algorithm.

\svs\noindent \tb{Query Phase 1:} For a key generation query, $\mc{B}_2$ generates
$(USK_j, UPK_j)$ by simply running the key generation algorithm. It gives
$UPK_j$ to $\mc{A}$.
For a key corruption query, it simply responds with $USK_j$ to $\mc{A}$.
For a decryption query on $(S, CH = (CM, \sigma, VK), i, AU)$, it first verify
the validity of $\sigma$ by running the verification algorithm of OTS.
If $VK = VK^*$, then it sets $\ts{Bad} = 1$ and outputs a forgery $(\sigma,
CM)$. Otherwise, it responds with a session key by simply running the decryption
algorithm.

\svs\noindent \tb{Challenge:} To create a ciphertext $(CH^*, CK^*)$,
$\mc{B}_2$ simply runs the encapsulation algorithm except that it uses the
signing oracle of OTS to create $\sigma^*$. It sets $CK_0^* = CK^*$ and
$CK_1^* = RK$ by selecting a random $RK$. It flips a coin $\mu \in \bits$ and
gives $(CH^*, CK_\mu^*)$ to $\mc{A}$.

\svs \noindent \tb{Query Phase 2}: The decryption queries are handled as the
same as query phase 1.

\svs\noindent \tb{Guess:} $\mc{B}_2$ outputs $\perp$.

By combining the analysis of two simulators, we obtain the following equation
	\begin{align*}
	\Adv_{\mc{A}}^{G_0} - \Adv_{\mc{A}}^{G_1}
	 \leq \Adv_{DBE}^{SS-CCA}(\lambda) + \Adv_{OTS}^{SUF}(\lambda).
	\end{align*}
This completes our proof.
\end{proof}

\begin{lemma} \label{lem:dbe-adcca-g1-g2}
If the DBE$_{SS}$ scheme is SS-CCA secure and the OTS scheme is SUF secure,
then no PPT adversary can distinguish $\tb{G}_1$ from $\tb{G}_2$ with a
non-negligible advantage.
\end{lemma}

\begin{proof}
The proof of this lemma is almost the same as that of Lemma
\ref{lem:dbe-adcca-g0-g1} except that a simulator $\mc{B}_1$ that breaks the
DBE$_{SS}$ scheme defines the initial set $\tilde{S}_{SS} = \{ 2j - (1-z_j^*)
\}_{j=1}^L$ in the setup phase. Because of this change, the key generation,
key corruption, decryption queries are slightly changed. In the challenge
step, the set $S_1^*$ become the challenge set. The detailed description of
the challenge setup is given as follows:

\svs\noindent \tb{Challenge:} $\mc{A}$ submits a challenge set $S^* \subseteq
[L]$ such that $S^* \subseteq KQ \setminus CQ$, $\mc{B}$ proceeds as follows:
	\begin{enumerate}
	\item It defines $S_0^* = \{ 2j - z_j^* \}_{j \in S^*}$ and $S_1^* =
	\{ 2j - (1-z_j^*) \}_{j \in S^*}$.
	It obtains $(CH_{SS,0}^*, CK_{SS,0}^*)$ by running $\tb{DBE}_{SS}.\tb{Encaps}
	(S_0^*, \{ (k, UPK_{SS,k}) \}_{k \in S_0^*}, PP_{SS}, VK^*)$.
	It submits a challenge set $S_1^*$ to $\mc{C}$ and receives $(CH_{SS,1}^*,
	CK_{SS,1}^*)$.
	It selects a random $RK_{SS,0}$.

	\item It selects a session key $CK^*$. It obtains $CT_0^*$ and $CT_1^*$
	by running $\tb{SKE.Encrypt}(RK_{SS,0}, CK^*)$ and $\tb{SKE.Encrypt}
	(CK_{SS,1}^*, CK^*)$ respectively.

	\item It sets $CM^* = (CH_{SS,0}^*, CH_{SS,1}^*, CT_0^*, CT_1^*, z^*)$
	and calculates $\sigma^*$ by running $\tb{OTS.Sign} (SK^*, CM^*)$.

	\item It sets $CH^* = (CM^*, \sigma^*, VK^*)$, $CK_0^* = CK^*$, and
	$CK_1^* = RK$ by selecting a random $RK$. Next, it flips a random coin
	$\mu \in \bits$ and gives $(CH^*, CK_\mu^*)$ to $\mc{A}$
	\end{enumerate}

The description of a simulator $\mc{B}_2$ that breaks the OTS scheme is the
same as that of Lemma \ref{lem:dbe-adcca-g0-g1}. We omit the description of
this simulator.

Similar to the analysis of Lemma \ref{lem:dbe-adcca-g0-g1}, we obtain the
following equation
	\begin{align*}
	\Adv_{\mc{A}}^{G_1} - \Adv_{\mc{A}}^{G_2}
	&\leq \Adv_{DBE}^{SS-CCA}(\lambda) + \Adv_{OTS}^{SUF}(\lambda).
	\end{align*}
This completes our proof.
\end{proof}

\begin{lemma} \label{lem:dbe-adcca-g2-g3}
If the SKE scheme is OMI secure, then no PPT adversary can distinguish
$\tb{G}_2$ from $\tb{G}_3$ with a non-negligible advantage.
\end{lemma}

\begin{proof}
Suppose there exists an adversary $\mc{A}$ that distinguishes $\tb{G}_2$ from
$\tb{G}_3$ with a non-negligible advantage. Let $\mc{C}$ be a challenger of
the SKE scheme. A simulator $\mc{B}$ that breaks the SKE scheme by using
$\mc{A}$ with an input $1^L$ is described as follows:

\vs\noindent \tb{Setup:} In this step, $\mc{B}$ simply runs the normal setup
algorithm.

\svs\noindent \tb{Query Phase 1:} For the key generation, key corruption, and
decryption queries of $\mc{A}$, $\mc{B}$ simply handles these queries by
running the normal algorithms of DBE.

\svs\noindent \tb{Challenge:} $\mc{A}$ submits a challenge set $S^* \subseteq
[L]$, $\mc{B}$ proceeds as follows:
	\begin{enumerate}
	\item It generates $(SK^*, VK^*)$ by running $\tb{OTS.GenKey}(1^\lambda)$.

	\item It defines $S_0^* = \{ 2j - z_j^* \}_{j \in S^*}$ and $S_1^* =
	\{ 2j - (1-z_j^*) \}_{j \in S^*}$.
	It obtains $(CH_{SS,0}^*, CK_{SS,0}^*)$ by running $\tb{DBE}_{SS}.\tb{Encaps}
	(S_0^*, \{ (k, UPK_{SS,k}) \}_{k \in S_0^*}, PP_{SS}, VK^*)$.
	It also obtains $(CH_{SS,1}^*, CK_{SS,1}^*)$ by running $\tb{DBE}_{SS}.
	\lb\tb{Encaps} (S_1^*, \{ (k, UPK_{SS,k}) \}_{k \in S_1^*}, PP_{SS}, VK^*)$.
	It selects a random $RK_{SS,1}$.

	\item It selects a session key $CK^*$. It also selects a random $RK_0$.
	It submits $(CK^*, RK_0)$ to $\mc{C}$ and receives $CT_0^*$.
	It obtains $CT_1^*$ by running $\tb{SKE.Encrypt} (RK_{SS,1}, CK^*)$.

	\item It sets $CM^* = (CH_{SS,0}^*, CH_{SS,1}^*, CT_0^*, CT_1^*, z^*)$
	and calculates $\sigma^*$ by running $\tb{OTS.Sign} (SK^*, CM^*)$.

	\item It sets $CH^* = (CM^*, \sigma^*, VK^*)$, $CK_0^* = CK^*$, and
	$CK_1^* = RK$ by selecting a random $RK$.
	Next, it flips a random coin $\mu \in \bits$ and gives $(CH^*, CK_\mu^*)$
	to $\mc{A}$
	\end{enumerate}

\svs\noindent \tb{Query Phase 2:} For a decryption query, $\mc{B}$ handles
this decryption query as the same as query phase 1.

\svs\noindent \tb{Guess:} $\mc{A}$ outputs a guess $\mu'$. If $\mu = \mu'$,
then $\mc{B}$ outputs 1. Otherwise, it outputs 0.
\end{proof}

\begin{lemma} \label{lem:dbe-adcca-g3-g4}
If the SKE scheme is OMI secure, then no PPT adversary can distinguish
$\tb{G}_3$ from $\tb{G}_4$ with a non-negligible advantage.
\end{lemma}

\begin{proof}
The proof of this lemma is almost the same as that of Lemma
\ref{lem:dbe-adcca-g2-g3} except the generation of the challenge ciphertext.
The detailed description of the challenge setup is given as follows:

\svs\noindent \tb{Challenge:} $\mc{A}$ submits a challenge set $S^* \subseteq
[L]$, $\mc{B}$ proceeds as follows:
	\begin{enumerate}
	\item It generates $(SK^*, VK^*)$ by running $\tb{OTS.GenKey}(1^\lambda)$.

	\item It defines $S_0^* = \{ 2j - z_j^* \}_{j \in S^*}$ and $S_1^* =
	\{ 2j - (1-z_j^*) \}_{j \in S^*}$.
	It obtains $(CH_{SS,0}^*, CK_{SS,0}^*)$ by running $\tb{DBE}_{SS}.\tb{Encaps}
	(S_0^*, \{ (k, UPK_{SS,k}) \}_{k \in S_0^*}, PP_{SS}, VK^*)$.
	It also obtains $(CH_{SS,1}^*, CK_{SS,1}^*)$ by running $\tb{DBE}_{SS}.
	\lb\tb{Encaps} (S_1^*, \{ (k, UPK_{SS,k}) \}_{k \in S_1^*}, PP_{SS}, VK^*)$.
	It selects a random $RK_{SS,0}$.

	\item It selects a session key $CK^*$. It also selects random $RK_0$ and
	$RK_1$.
	It obtains $CT_0^*$ by running $\tb{SKE.Encrypt} \lb(RK_{SS,0}, RK_0)$.
	It submits $(CK^*, RK_1)$ to $\mc{C}$ and receives $CT_1^*$.

	\item It sets $CM^* = (CH_{SS,0}^*, CH_{SS,1}^*, CT_0^*, CT_1^*, z^*)$
	and calculates $\sigma^*$ by running $\tb{OTS.Sign} (SK^*, CM^*)$.

	\item It sets $CH^* = (CM^*, \sigma^*, VK^*)$, $CK_0^* = CK^*$, and
	$CK_1^* = RK$ by selecting a random $RK$.
	Next, it flips a random coin $\mu \in \bits$ and gives $(CH^*, CK_\mu^*)$
	to $\mc{A}$
	\end{enumerate}
This completes our proof.
\end{proof}

\section{Conclusion}

In this paper, we proposed an efficient DBE scheme in bilinear groups and
proved its adaptive CCA security under the $q$-Type assumption. Our adaptive
CCA secure DBE scheme has constant size ciphertexts, constant size private
keys, and linear size public keys. The public key verification of our DBE
scheme requires only a constant number of pairing operations and the linear
number of efficient group membership operations. An interesting problem is to
design an efficient and adaptive CCA secure DBE scheme under standard
assumptions weaker than the $q$-Type assumption.

\bibliographystyle{plain}
\bibliography{dbe-for-adcca-security}

\end{document}